\documentclass[letterpaper, 10 pt, conference]{ieeeconf}
\IEEEoverridecommandlockouts                    
\usepackage{graphicx} 
\usepackage{epsfig} 
\usepackage{mathptmx} 
\usepackage{times} 
\usepackage{amsmath} 
\usepackage{amssymb, bbm}  
\usepackage{todonotes}
\usepackage{xcolor}
\usepackage{multicol}
\usepackage[bb=boondox]{mathalfa}
\usepackage{subcaption}
\usepackage{enumerate}

\newtheorem{assumption}{Assumption}
\newtheorem{theorem}{Theorem}

\newtheorem{lemma}[theorem]{Lemma}
\usepackage[utf8]{inputenc}
\usepackage[english]{babel}

\title{\LARGE \bf
Accelerated consensus in multi-agent networks via \\ memory of local averages}

\author{Aditya Bhaskar$^{*}$, Shriya Rangarajan$^{*}$, Vikram Shree$^{*}$, Mark Campbell, and Francesca Parise
\thanks{*Equal contribution. }
\thanks{A. Bhaskar, V. Shree, and M. Campbell are with the Sibley School of Mechanical and Aerospace Engineering, Cornell University, Ithaca, NY, USA. Email: {\tt\small \{ab2823, vs476, mc288\}@cornell.edu}}
\thanks{S. Rangarajan is with the Department of City and Regional Planning, Cornell University, Ithaca, NY, USA. Email: {\tt\small sr2248@cornell.edu}}
\thanks{F. Parise is with the School of Electrical and Computer Engineering, Cornell University, Ithaca, NY, USA. Email: {\tt\small fp264@cornell.edu}}
\thanks{A. Bhaskar was supported by the National Science Foundation (NSF), United States under grant number CMMI-1634664.
V. Shree was supported by the NRI program of NSF, United States under Award \#1830497.}
}

\begin{document}

\maketitle
\thispagestyle{empty}
\pagestyle{empty}

\begin{abstract}
Classical mathematical models of information sharing and updating in multi-agent networks use linear operators.
In the paradigmatic DeGroot model, agents update their states with linear combinations of their neighbors’ current states. 
In prior work, an accelerated averaging model employing the use of memory has been suggested to accelerate convergence to a consensus state for undirected networks. 
There, the DeGroot update on the current states is followed by a linear combination with the previous states. We propose a modification where the DeGroot update is applied to the current and previous states and is then followed by a linear combination step. We show that this simple modification applied to undirected networks permits convergence even for periodic networks. Further, it allows for faster convergence than the DeGroot and accelerated averaging models for suitable networks and model parameters.\end{abstract}


\section{Introduction}

Linear models for information sharing play a crucial role in multi-agent systems and sensor networks. Among these, distributed averaging algorithms have been intensively studied as a way to reach global consensus with local computations. One of the first such models was proposed by DeGroot \cite{degroot1974reaching}.
There, agents update their state by taking the weighted average of their neighbours' states at each time step. Many variations have since been studied to account for real-world phenomena, such as changing communication patterns, or individuals' stubbornness, and to optimize for the rate of convergence. 

As a starting point for this paper, we consider an accelerated averaging model first proposed by Muthukrishnan et al. in \cite{muthukrishnan1998first}. Therein, authors suggest a modification of the classic DeGroot scheme where agents update their states by first taking a DeGroot update of the current states and then performing a linear combination between this and their previous states, as shown in Fig.~\ref{fig:MuthuModelIntuit}. In \cite{muthukrishnan1998first} they show that this additional memory step allows for faster convergence in networks where the original DeGroot model also converges. 
In this paper, our modification of the Muthukrishnan et al. averaging scheme involves the agents performing DeGroot averaging on the current and previous states followed by a linear combination step, as shown in Fig.~\ref{fig:PropoModelIntuit}. Hereafter, we call this model the Memory of Local Averages (MLA) model. Besides deriving convergence guarantees for this modified scheme, our main objective is to show that this simple modification leads to two important consequences: First, we demonstrate conditions for undirected, connected and periodic networks under which MLA achieves consensus despite the DeGroot and the accelerated averaging models failing to do so. Second, we give sufficient conditions under which MLA achieves faster convergence than the DeGroot model and the accelerated averaging algorithm.

\begin{figure}
     \centering
     \begin{subfigure}[b]{0.235\textwidth}
         \centering
         \includegraphics[width=\textwidth]{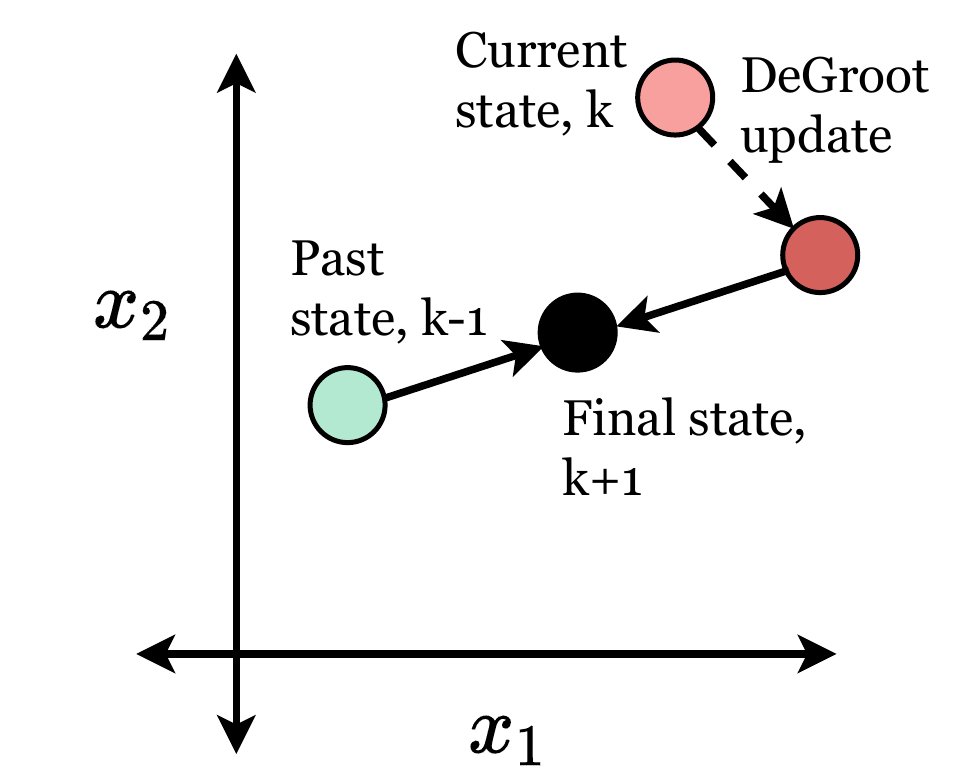}
         \caption{\small Accelerated averaging model}
         \label{fig:MuthuModelIntuit}
     \end{subfigure}
     \begin{subfigure}[b]{0.235\textwidth}
         \centering
         \includegraphics[width=\textwidth]{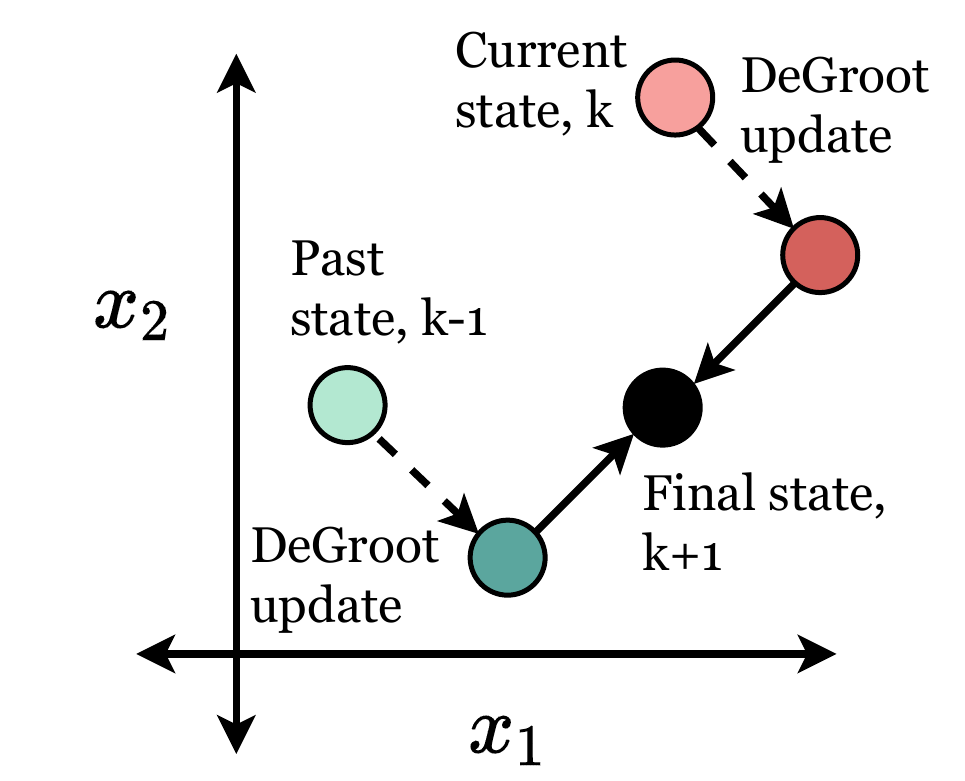}
         \caption{\small MLA model}
         \label{fig:PropoModelIntuit}
     \end{subfigure}
        \caption{\small Representation of discrete-time averaging algorithms in a two-dimensional state space $(x_1, x_2)$. (a) In the accelerated averaging model, a DeGroot update is performed on current state at $k$. This is followed by a linear combination step to get the final state at $k+1$. (b) In MLA, we apply Degroot update on the current and past states. This is followed by a linear combination of the intermediate states to obtain the final state at $k+1$.}
        \label{fig:ModelIntuit}
        \vspace{-0.2in}
\end{figure}

Our work is part of a rich literature on opinion dynamics, information sharing and consensus models. Besides classic DeGroot dynamics, several modifications have been proposed to include stubborn agents (e.g., \cite{acemouglu2013opinion}, \cite{friedkin1990social}) or dynamics on random or state-dependent networks (e.g. in \cite{deffuant2000mixing} \cite{hegselmann2002opinion}) among others. We refer to \cite{FB-LNS} and \cite{olfati2007consensus} for recent surveys. 

Our work is most related to the subsection of this literature that focuses on capturing convergence rates while maintaining the original DeGroot structure. In general, faster convergence can be achieved either by tuning the network weights as suggested in e.g. \cite{xiao2004fast} \cite{xiao2007distributed} or by modifying the update dynamics as suggested in e.g. \cite{muthukrishnan1998first}. Our work belongs to the latter class. We note that whereas in \cite{muthukrishnan1998first} and our work, faster convergence is obtained by introducing one step of memory, the use of a prediction step has also been investigated (see \cite{aysal2008accelerated, wang2013accelerated}). Similar to these works, we retain linearity of the update equation which permits for analysis using techniques of linear algebra and Markov chains. We conclude by noting that significant work focused also on deriving finite-time consensus schemes
\cite{wang2010finite, li2011finite, cao2014finite}. Such schemes however require non-linear updates, which may carry additional computational burden as compared to their linear counterparts.

\section{The MLA Model}
\subsection{Model formulation}
Let $n \geq 2$ agents be connected in a network with corresponding weighted adjacency matrix, $A = [a_{ij}]_{n\times n}$. Here, $a_{ij}$ denotes the weight that agent $i$ assigns to agent $j$. The state of each agent, $i$, at time step, $k \in \mathbb{Z}_{\geq 0}$, is denoted by $x_{i}(k)$. Throughout the paper, we use the following assumption.
\begin{assumption}
\label{NonegRowStochastic}
The weighted adjacency matrix $A_{n \times n}$ is \textit{non-negative} i.e. $a_{ij}\geq0 \ \forall \ i,j \in \{1, \hdots, n\}$ and \textit{row-stochastic} i.e. $\sum_{j=1}^{n} a_{ij} = 1 \ \forall \ i \in \{1,\hdots, n\}$.
\end{assumption}
In the DeGroot model, each agent updates its state with the weighted average of the current states of its local neighbors,
\begin{gather}
\nonumber
x_{i, \text{DG}}(k+1)= \sum_{j=1}^{n}a_{ij}x_j(k), \\
\Rightarrow \textbf{x}_{\text{DG}}(k+1) = A\textbf{x}(k), \label{DeGroot2}
\end{gather}
where $\textbf{x}(k) = [x_{1}(k), \hdots , x_{n}(k)]^{T} \in \mathbb{R}^{n}$.

In contrast, the MLA model is an averaging algorithm with memory which takes into account each agent's current and \textit{previous} states. In particular, given the current states of the agents $\textbf{x}(k)$, and the previous states $\textbf{x}(k-1)$, we first apply Degroot model to get intermediate states $\textbf{z}(k)$ and $\textbf{z}(k-1)$, respectively, as shown below.
\begin{gather}
\nonumber
\textbf{z}(k)= A\textbf{x}(k) \ \text{and} \ \textbf{z}(k-1) = A \textbf{x}(k-1). \label{LinExtrapol}
\end{gather}
We obtain the final states $\textbf{x}(k+1)$ by taking a linear combination of $\textbf{z}(k)$ and $\textbf{z}(k-1)$, with weighting parameter $\gamma$.
\begin{gather}
\nonumber
\textbf{x}(k+1) = \gamma \textbf{z}(k) + (1-\gamma)\textbf{z}(k-1)\\
\Rightarrow \textbf{x}(k+1) = \gamma A \textbf{x}(k) +(1-\gamma) A \textbf{x}(k-1). \label{Model4}
\end{gather}
In MLA, each agent need only store its local average at the previous time step instead of the true states of its neighbors. Note that for $\gamma = 1$, Eq.~(\ref{Model4}) reduces to the DeGroot model.

\subsection{Augmented system}
To study convergence properties of system given by Eq.~\ref{Model4} we define an augmented state, $\hat{\textbf{x}}(k) \in \mathbb{R}^{2n}$ and an augmented iteration matrix, $\hat{A} \in \mathbb{R}^{2n \times 2n}$, as follows 
\begin{gather}
\hat{\textbf{x}}(k) \equiv \begin{bmatrix}
       \textbf{x}(k) \\[0.5em]
       \textbf{x}(k-1) 
     \end{bmatrix}, \label{xhat1} \\
\Rightarrow \hat{\textbf{x}}(k+1) = \begin{bmatrix}\begin{array}{c|c}
       \gamma A & (1-\gamma) A \\[0.5em]
       \hline
       \mathbb{I}_{n\times n} & \mathbb{0}_{n\times n} 
     \end{array}\end{bmatrix} \hat{\textbf{x}}(k) \equiv \hat{A}\hat{\textbf{x}}(k).\label{xhat3}
\end{gather}

Spectral properties of $\hat{A}$ in relation to $A$ are discussed in the following lemma.

\begin{lemma}
\label{Eigenpair}
If ($\lambda, \textbf{v}$) denotes an eigenpair of $A$, then the eigenpairs of $\hat{A}$ are ($\hat{\lambda}_{\pm}, \hat{\textbf{v}}_{\pm}$) where,
\begin{gather}
\hat{\lambda}^2 - \gamma \lambda \hat{\lambda} + (\gamma-1)\lambda = 0, \label{quadraticeqn} \\
\Rightarrow \hat{\lambda}_{\pm} = \frac{\gamma \lambda \pm \sqrt{\gamma^2 \lambda^2 -4 (\gamma-1) \lambda}}{2},  \label{quadroots} \\
\hat{\textbf{v}}_{\pm} = \begin{bmatrix}
      \hat{\lambda}_{\pm}\textbf{v} \\[0.5em]
       \textbf{v} 
     \end{bmatrix}.\label{Eigvec} 
\end{gather}
\end{lemma}

\begin{proof}
 We perform eigenanalysis on $\hat{A}$. Let an eigenvalue and corresponding right eigenvector be represented by $\hat{\lambda} \in \mathbb{C}$ and $\hat{\mathbf{v}} \in \mathbb{C}^{2n}$  respectively, s.t.,
\begin{gather}
\hat{A}\hat{\mathbf{v}} = \hat{\lambda}\hat{\mathbf{v}}.\label{Eig1}
\end{gather}
Define $\mathbf{v}_\text{Top}, \mathbf{v}_\text{Bottom} \in \mathbb{C}^{n}$ s.t. $\hat{\mathbf{v}} = [\mathbf{v}_\text{Top}; \mathbf{v}_\text{Bottom}]$. Substituting the expressions for $\hat{A}$ and $\hat{\mathbf{v}}$ in Eq.~(\ref{Eig1}), we get, 
\begin{gather}
\nonumber
\begin{bmatrix}\begin{array}{c|c}
       \gamma A & (1-\gamma) A \\[0.5em]
       \hline
       \mathbb{I}_{n \times n} & \mathbb{0}_{n \times n} 
     \end{array}\end{bmatrix} 
     \begin{bmatrix}
       \mathbf{v}_\text{Top} \\[0.5em]
       \mathbf{v}_\text{Bottom}  
     \end{bmatrix} = \hat{\lambda} \begin{bmatrix}
       \mathbf{v}_\text{Top} \\[0.5em]
       \mathbf{v}_\text{Bottom}  
     \end{bmatrix}, \\
\Rightarrow \begin{bmatrix}
       \gamma A \mathbf{v}_\text{Top} +(1-\gamma)A \mathbf{v}_\text{Bottom} \\[0.5em]
       \mathbf{v}_\text{Top} 
     \end{bmatrix} = \begin{bmatrix}
      \hat{\lambda} \mathbf{v}_\text{Top} \\[0.5em]
       \hat{\lambda} \mathbf{v}_\text{Bottom}
     \end{bmatrix}. \label{Eig2}    
\end{gather}
Solving the system of Eq.~(\ref{Eig2}), we get,
\begin{gather} 
\nonumber
\gamma \hat{\lambda} A \mathbf{v}_\text{Bottom} +(1-\gamma) A \mathbf{v}_\text{Bottom} = \hat{\lambda}^2 \mathbf{v}_\text{Bottom}, \\
\Rightarrow A \mathbf{v}_\text{Bottom} = \left(\frac{ \hat{\lambda}^2}{\gamma \hat{\lambda} +1-\gamma}\right) \mathbf{v}_\text{Bottom}. \label{aHatEigEq}
\end{gather}
Let $(\lambda, \mathbf{v})$ denote an eigenpair for $A$. We get,
\begin{gather}
\lambda = \frac{ \hat{\lambda}^2}{\gamma \hat{\lambda} +1-\gamma},\label{aHatEig}\\
\mathbf{v}_\text{Bottom} = \mathbf{v}. \label{vbottom}
\end{gather}

Eq.~(\ref{aHatEig}) can be rewritten to obtain the eigenvalue $\hat{\lambda}$ as follows,
\begin{gather}
\hat{\lambda}^2 - \gamma \lambda \hat{\lambda} + (\gamma-1)\lambda = 0, \\
\Rightarrow \hat{\lambda}_{\pm} = \frac{\gamma\lambda \pm \sqrt{\gamma^2 \lambda^2 -4 (\gamma-1) \lambda}}{2}.
\end{gather}

Substituting Eq.~(\ref{vbottom}) in Eq.~(\ref{Eig2}), we get,
\begin{gather}
\mathbf{v}_\text{Top} = \hat{\lambda} \mathbf{v}. \label{vtop}
\end{gather}

Combining Eq.~(\ref{vbottom}) and (\ref{vtop}), we get the eigenvector $\hat{\textbf{v}}$ as follows,
 \begin{gather}
\hat{\textbf{v}}_{\pm} = \begin{bmatrix}
      \hat{\lambda}_{\pm}\textbf{v} \\[0.5em]
       \textbf{v} 
     \end{bmatrix}.
\end{gather}
\end{proof}

\section{Convergence Guarantees and Consensus Value}
\label{convergenceCriteriaSec}

\subsection{Convergence guarantees for connected undirected networks }

Convergence properties of (\ref{xhat3}) (and thus (\ref{Model4})) can be connected to spectral properties of $\hat{A}$ as stated in Lemma~\ref{SemiConvCriteria}. 

\begin{lemma}[Theorem 2.7 in \cite{FB-LNS}]
\label{SemiConvCriteria}
A square matrix, $B \in \mathbb{R}^{n \times n}$, is semi-convergent and not convergent i.e. $\lim_{k \to \infty}B^k$ exists different from $\mathbb{0}_{n\times n}$ if and only if (i) 1 is an eigenvalue of $B$ and is semi-simple (ii) all other eigenvalues of $B$ have magnitude strictly less than 1.
\end{lemma}

\begin{assumption}
\label{irreducibleAssume}
The weighted adjacency matrix $A$ is \textit{symmetric} i.e. $a_{ij}=a_{ji} \ \forall \ i,j \in \{1,2,\hdots,n\}$ and \textit{irreducible} i.e. $\sum_{k=0}^{n-1} A^{k} >0$.
\end{assumption}

\textit{Remark:} Assumption \ref{irreducibleAssume} is equivalent to the network being undirected and connected. We define the spectrum of $A$ as the set of all eigenvalues of $A$ i.e. $\text{spec}(A) \equiv \{\lambda_1, \lambda_2,\hdots,\lambda_n\}$ with the condition $\lambda_1 \geq \lambda_2 \geq \hdots \geq \lambda_n$ and the corresponding eigenvectors to be $\{\textbf{v}_1, \textbf{v}_2,\hdots,\textbf{v}_n\}$. Due to Perron-Frobenius theorem for irreducible matrices, under Assumption~\ref{irreducibleAssume}, (i) $\lambda_1 = 1$ is a simple eigenvalue of $A$, (ii) $|\lambda_i|\leq1 \ \forall \ i \in \{2,3,\hdots,n\}$ such that $1 = \lambda_1 > \lambda_2 \geq \hdots \geq \lambda_n \geq -1$.
\begin{theorem}
\label{irreducibleConvCriteria}
 Suppose Assumptions \ref{NonegRowStochastic} and \ref{irreducibleAssume} hold. Then $\hat{A}$ is semi-convergent if and only if:
 \begin{enumerate}[(i)]
     \item $\gamma \in (0,2)$,
     \item $2\gamma\lambda_{n}-\lambda_{n}+1 > 0$, 
 \end{enumerate}
 where $\lambda_{n}$ is the smallest eigenvalue of A. 
 \end{theorem}
 
 \begin{proof}
 The proof has three parts shown below:

 (1) $\hat{A}$ is semi-convergent $\Rightarrow$ $\gamma\in (0,2)$
     
     We prove this by contradiction. Since $\lambda_1= 1$ is an eigenvalue of $A$, from Lemma \ref{Eigenpair}, $\{1, \gamma-1\}$ are eigenvalues for $\hat{A}$ with corresponding eigenvectors $\{[\mathbf{v}_1 ; \mathbf{v}_1], [(\gamma-1) \mathbf{v}_1; \mathbf{v}_1]\}$. Now if $\gamma \in \mathbb{R}\setminus(0,2]$ then $\hat{A}$ has an eigenvalue $\hat{\lambda}_{1-}=\gamma-1 \in \mathbb{R}\setminus(-1,1]$, which is absurd from Lemma \ref{SemiConvCriteria} since $\hat{A}$ is semi-convergent.

     If $\gamma=2$, $\hat{A}$ has an eigenvalue $\hat{\lambda}=1$ with algebraic multiplicity $2$ (from Lemma \ref{Eigenpair}). The associated eigenspace is $E_{1} = \{\hat{\mathbf{v}}| \ \hat{A}\hat{\mathbf{v}} = \hat{\mathbf{v}} \}$. Note that for $\gamma=2$,
\begin{gather}
\hat{A}\hat{\mathbf{v}} = \begin{bmatrix}\begin{array}{c|c}
       2A & -A \\[0.5em]
       \hline
       \mathbb{I}_{n\times n} & \mathbb{0}_{n\times n} 
     \end{array}\end{bmatrix} 
     \begin{bmatrix}
       \hat{\mathbf{v}}_\text{1} \\[0.5em]
       \hat{\mathbf{v}}_\text{2}  
     \end{bmatrix} = 
          \begin{bmatrix}
       \hat{\mathbf{v}}_\text{1} \\[0.5em]
       \hat{\mathbf{v}}_\text{2}  
     \end{bmatrix},
\end{gather}
     which implies $\hat{\mathbf{v}}_{1} = \hat{\mathbf{v}}_{2} = \mathbf{v}$. Thus, $2A\hat{\textbf{v}}_1 - A\hat{\textbf{v}}_2 = A\textbf{v} = \hat{\textbf{v}}_1 = \textbf{v}$. Here $\textbf{v}$ is an eigenvector of A. Since by assumption, the geometric multiplicity of $\lambda=1$ in A is 1, the geometric multiplicity of $\hat{\lambda}=1$ in $\hat{A}$ is also 1. This is absurd since $\hat{A}$ semi-convergent implies $\hat{\lambda}=1$ has same algebraic and geometric multiplicity.
     
    (2) $\hat{A}$ is semi-convergent $\Rightarrow$ $2\gamma\lambda_{n}-\lambda_{n}+1 > 0$

    By Eq.~(\ref{quadraticeqn}) $\text{spec}(\hat{A}) = \{ \hat{\lambda}_{i\pm}| \ \lambda_{i} \in \text{spec}(A) \}$. Since $\hat{A}$ is semi-convergent $|\hat{\lambda}_{i\pm}| \leq 1 \ \forall \ i$. Moreover, since $\hat{\lambda} = 1$ has geometric multiplicity 1 and $\hat{A}$ is semi-convergent, it must be that $\hat{\lambda}_{1}$ has geometric multiplicity 1. Given that $\hat{\lambda}_{1+} = 1$, we conclude that $|\hat{\lambda}_{i\pm}|<1 \ \forall \ i\in \{2,3,\hdots,n\}$. 
    
\begin{lemma}[Lemma 8.5 in \cite{ren2010distributed}]
The polynomial $z^2 + az+ b=0$, where $a, b \in \mathbb{C}$, has roots $|z_{1,2}|<1$ if and only if roots $s_1$ and $s_2$ of $(1+a+b)s^2+2(1-b)s+b-a+1=0$ satisfy $\mathfrak{Re}(s_1) <0$ and $\mathfrak{Re}(s_2) <0$.
\label{LemmaRenCao}

\end{lemma}

    Applying Lemma \ref{LemmaRenCao} to Eq.~(\ref{quadraticeqn}),  $|\hat{\lambda}_{i\pm}| < 1$ if and only if the following Eq.~(\ref{Lemma}) has roots with strictly negative real parts.
    \begin{gather}
    \nonumber
    (1-\gamma\lambda_i+(\gamma-1) \lambda_i)s^2 + 2(1-(\gamma-1) \lambda_i)s + \hdots\\
    \nonumber
    \hdots + ((\gamma-1) \lambda_i + \gamma\lambda_i +1) = 0, \\
    \Rightarrow (1- \lambda_i)s^2 + 2(1-\gamma \lambda_i+\lambda_i)s + (2\gamma\lambda_{i}-\lambda_{i}+1) = 0, \label{Lemma}
    \end{gather}

    For fixed $i \in \{2,3,\hdots,n\}$, let the roots of Eq.~(\ref{Lemma}) be $s_1, s_2 \in \mathbb{C}$. We have the sum of roots of Eq.~(\ref{Lemma}),
    \begin{gather}
    s_1 + s_2 = \frac{-2(1-\gamma \lambda_i+\lambda_i)}{1-\lambda_i} \in \mathbb{R}. \label{sumroots}
    \end{gather}
    Since the sum of roots is real, $s_1$ and $s_2$ are of the form, $s_1 = a_1 + bi$ and $s_2 = a_2 - bi$, where $a_1, a_2, b \in \mathbb{R}$. We then note that,
    \begin{gather}
    a_1 < 0, \ a_2 < 0 \Leftrightarrow c_1: a_1 + a_2 < 0, \ c_2: a_1a_2>0. \label{negrealpart}
    \end{gather}
    We study conditions $c_1$ and $c_2$ independently. Note that for condition $c_1$,
    \begin{gather}
    a_1 + a_2 =  s_1 + s_2 = \frac{-2(1-(\gamma-1) \lambda_i)}{1-\lambda_i} < 0 
    \Leftrightarrow (\gamma-1) \lambda_i < 1. 
    \label{convcrit1}
    \end{gather}
    Since, as proved above,  semi-convergence implies $\gamma \in (0,2)$ and $|\lambda_i| \leq 1$, condition $c_1$ is automatically satisfied $\forall \ i \in \{2,3,\hdots,n\}$. For condition $c_2$, when $b=0$, using the product of roots of Eq.~(\ref{Lemma}), we can write,
    \begin{gather}
    \nonumber
    a_1a_2 =  s_1s_2 = \frac{2\gamma\lambda_{i}-\lambda_{i}+1}{1-\lambda_i} > 0, \\
    \nonumber
    \Leftrightarrow 2\gamma\lambda_{i}-\lambda_{i}+1 > 0, \ 
    \forall \ i \in \{2, \hdots, n\}.
    \end{gather}
    Specifically, for $\lambda_n$, we get,
    \begin{gather}
    2\gamma\lambda_{n}-\lambda_{n}+1  > 0.
    \label{convcrit2}
    \end{gather}    
    When $b\neq0$, $a_1=a_2 \ \Rightarrow a_1 a_2 = a_{1}^{2}\ \geq\ 0$. Since, $a_1 + a_2 < 0$ from condition $c_1$, $a_1\neq0$. Hence, $a_1 a_2 >0$. 
    
(3) The conditions (i) $\gamma \in (0,2)$ and (ii) $2\gamma\lambda_{n}-\lambda_{n}+1 > 0$ $\Rightarrow$ $\hat{A}$ is semi-convergent.

Recall that from Lemma \ref{Eigenpair}, the eigenvalue $\lambda$ of $A$ gets mapped to eigenvalues $\hat{\lambda}_{\pm}$ in $\hat{A}$. Then $\lambda=1$ is mapped to $\{1,  \gamma -1\}$, where $\gamma \in (0, 2)$.

We next show that all the other eigenvalues of $\hat{A}$ are mapped inside a unit disk, which is enough to prove that $\hat{A}$ is semi-convergent by Lemma~\ref{SemiConvCriteria}. Again, applying Lemma \ref{LemmaRenCao} to Eq.~(\ref{quadraticeqn}), we know that $\hat{\lambda}_{i\pm}$ lies within the unit-disk if and only if the roots of Eq. (\ref{Lemma}) lie in the open left half plane. As before, the roots of Eq. (\ref{Lemma}) are in the left half plane if and only if,
\begin{gather}
\label{lhpCriteria_c1}
c_1: -2(1-(\gamma-1) \lambda_i) < 0, \\ 
c_2: (2\gamma-1)\lambda_i+1 > 0.
\label{lhpCriteria_c2}
\end{gather} 

Recall that $\lambda_{i}\in[-1,1)$ for $i \in \{2, \hdots, n\}$, hence for $\gamma \in (0,2)$, criterion $c_1$ holds. For proving $c_2$, consider two cases: (a) $\gamma \in (0, 1)$ and (b) $\gamma \in [1,2)$. 

For case (a), $|2\gamma - 1| < 1$, and since $|\lambda_{i}| \leq 1$, the condition $(2\gamma-1)\lambda_i+1 > 0,$ holds $ \forall \ i \in \{2,\hdots, n\}$. 
For case (b), $2\gamma - 1 >0$ and $f(\lambda_i) \equiv (2\gamma-1)\lambda_i+1$ is a monotonically increasing function. Since, $\lambda_i \geq \lambda_n \ \forall \ i \in \{2,\hdots, n\}$ and $f(\lambda_{n})>0$, we infer that $f(\lambda_{i}) = (2\gamma-1)\lambda_i+1 \geq f(\lambda_{n}) > 0 \ \forall \ i \in \{2,\hdots, n\}$.

Overall, the eigenvalues of $\hat{A}$ are $\{1, \gamma-1 , \hat{\lambda}_{i\pm}\}$, with $\hat{\lambda}_{1+}=1$ simple and $(\gamma-1), \ \hat{\lambda}_{i\pm}$ inside the unit disk $\forall \ i \in \{2, \hdots, n\}$. From Lemma \ref{SemiConvCriteria} $\hat{A}$ is semi-convergent. 
\end{proof}

\subsection{Special case of periodic networks}

Note that Theorem \ref{irreducibleConvCriteria} does not require the network to be aperiodic. If the network is periodic, we get $\lambda_n = -1$ and convergence can be guaranteed by Theorem \ref{irreducibleConvCriteria} for $\gamma \in (0, 1)$. This is in sharp contrast with the DeGroot model where, if the network is periodic, we can always find an initial condition that causes persistent oscillations.

For comparison, we next discuss the accelerated averaging model by Muthukrishnan et al. \cite{muthukrishnan1998first}. Similar to our model, it also takes into account agents' current and previous states, however the update equation in \cite{muthukrishnan1998first} is given by,
\begin{gather}
\textbf{x}(k+1)= \beta A\textbf{x}(k) + (1-\beta) \textbf{x}(k-1), \label{muthuModel}
\end{gather}
where $k\in \mathbb{Z}_{\geq0}, \beta \in \mathbb{R}$ and initial states $\mathbf{x}(0) = \mathbf{x}(-1) = \mathbf{x}_{0}$.

\begin{lemma}
Under Assumptions~\ref{NonegRowStochastic} and~\ref{irreducibleAssume}, if the network associated with $A$ is periodic, there exist initial conditions $\mathbf{x}_0$ s.t. the accelerated averaging model given by Eq.~(\ref{muthuModel}) does not converge.
\end{lemma} 

\begin{proof}
 The augmented iteration matrix for the model in Eq.~(\ref{muthuModel}) is $A_{\beta} \in \mathbb{R}^{2n\times2n}$, 
\begin{gather}
     A_{\beta}  = 
\begin{bmatrix}\begin{array}{c|c}
       \beta A &  (1-\beta) \mathbb{I}_{n\times n} \\[0.5em]
       \hline
       \mathbb{I}_{n\times n} & \mathbb{0}_{n\times n} 
     \end{array}\end{bmatrix} \in \mathbb{R}^{2n \times 2n}.\label{muthuModelMat}
\end{gather}
Denote the eigenvalues of $A_{\beta}$ by $\mu_{i\pm}$, for $i \in \{1,2,\hdots,n\}$. As illustrated by Eq.~(11.13) in \cite{FB-LNS}, $\mu_{i\pm}$ can be obtained by performing eigenanalysis on $A_{\beta}$:
\begin{gather}
\mu_{i\pm} = \frac{\beta \lambda_{i} \pm \sqrt{\beta^{2} \lambda_{i}^{2} - 4\beta + 4}}{2}.\label{muthuModelEigs}
\end{gather}
The eigenvalues induced by $\lambda_{n} =-1$ are $\mu_{n\pm} = \{-1,\ 1 - \beta\}$. Thus, by Lemma \ref{SemiConvCriteria}, the model does not converge. E.g. Selecting an $\mathbf{x}_{0}$ along the eigenvector corresponding to 
$-1$ leads to persistent oscillations.
\end{proof}

We next numerically compare convergence properties of the MLA model with the DeGroot model and the accelerated averaging model. To do so, we consider a periodic ring network with four nodes, shown in Fig.~\ref{fig:convergenceForPeriodicGraphs}, and numerically simulate the system starting from 1000 randomly chosen initial conditions and plot the time series of the envelope of oscillations relative to the mean. It can be seen that for the chosen undirected, symmetric and periodic network, the DeGroot model as well as the accelerated averaging model oscillate whereas the MLA model reaches a consensus steady state.

\begin{figure}[!ht]
     \centering
        \includegraphics[ trim=60 10 40 30, clip, width=0.4\textwidth]{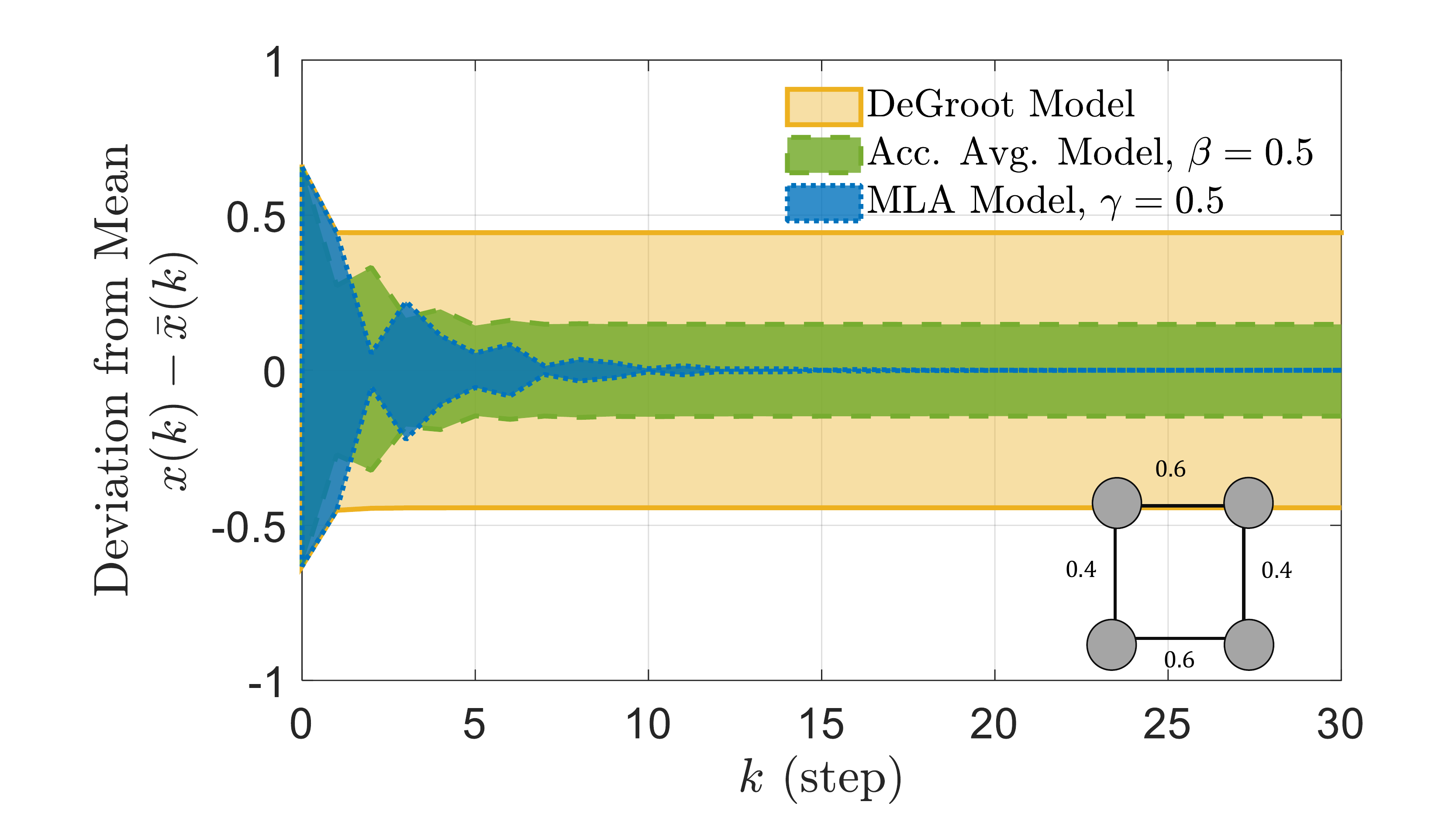}        

         \caption{\small
         Comparison of convergence in three averaging models for a periodic network. The plot shows the envelope enclosed between the maximum and minimum \textit{deviation} from the mean state and is simulated using 1000 randomly chosen initial conditions. In contrast to others, the MLA model converges to consensus.}
         \label{fig:convergenceForPeriodicGraphs}
\end{figure}

\subsection{Value of convergence}

\begin{lemma}
\label{consensus}
Under Assumptions \ref{NonegRowStochastic} and \ref{irreducibleAssume}, if the convergence criteria given in Theorem \ref{irreducibleConvCriteria} are satisfied, the MLA model given by Eq.~(\ref{Model4}), with initial conditions $\hat{\mathbf{x}}(0) = [\mathbf{x}(0); \mathbf{x}(0)]$, converges to a consensus value given by $\mathbf{w}_1^T\mathbf{x}(0)$, where $\mathbf{w}_1$ is left dominant eigenvector of $A$ corresponding to the dominant eigenvalue $\lambda_1=1$, normalized such that $\mathbf{w}_{1}^{T}\mathbb{1}_{n} = 1$.
\end{lemma}

\begin{proof}
 The final value of the states, $\hat{\textbf{x}}_\text{final}$, is,
\begin{gather}\label{finalVal}
\hat{\textbf{x}}_\text{final} = \lim_{k \to \infty}\hat{\textbf{x}}(k) = \lim_{k \to \infty} \hat{A}^k \hat{\textbf{x}}(0).
\end{gather}

Since $A$ has a simple eigenvalue at $\lambda_1 = 1$ and $\hat{A}$ is semi-convergent, Eq.~(\ref{quadraticeqn}) implies that $\hat{\lambda}_1 =1$ is a simple eigenvalue for $\hat{A}$ (see proof of Theorem~\ref{irreducibleConvCriteria}). Using the Jordan decomposition, we can write,
\begin{gather} \label{jordan}
\lim_{k \to \infty} \hat{A}^k = \lim_{k \to \infty} 
     T
     \begin{bmatrix}\begin{array}{c|c}
       1 & \mathbb{0}_{1 \times 2n-1} \\[0.5em]
       \hline
       \mathbb{0}_{2n-1 \times 1} & B_{2n-1\times 2n-1} 
     \end{array}\end{bmatrix}^{k} 
     T^{-1},
\end{gather}
where, $T = [\hat{\mathbf{v}}_{1\pm}, \hdots, \hat{\mathbf{v}}_{n\pm}]$, $T^{-1} = [\hat{\mathbf{w}}_{1\pm}^T; \hdots; \hat{\mathbf{w}}_{n\pm}^T]$, and $\hat{\mathbf{v}_{i\pm}}, \hat{\mathbf{w}_{i\pm}}$ denote the right and left eigenvectors of $\hat{A} \ \forall \ i\in \{1,\hdots,n\}$, such that $\hat{\mathbf{w}}_{i\pm}^{T}\hat{\mathbf{v}}_{i\pm} = 1, \ \forall \ i \in \{1,2,\hdots, n\}$. $B$ consists of Jordan blocks corresponding to eigenvalues, $\hat{\lambda}_{i\pm}$, satisfying $|\hat{\lambda}_{i\pm}|<1, \ \forall \ i\in \{2,3\hdots, n\}$ and $|\lambda_{i-}|<1$. Thus, we get,
\begin{gather} \label{jordanConverge}
\lim_{k \to \infty} \hat{A}^k = \lim_{k \to \infty} 
     T
     \begin{bmatrix}\begin{array}{c|c}
       1 & \mathbb{0}_{1 \times 2n-1} \\[0.5em]
       \hline
       \mathbb{0}_{2n-1 \times 1} & \mathbb{0}_{2n-1\times 2n-1} 
     \end{array}\end{bmatrix}
     T^{-1}.
\end{gather}
Combining Eqs.~(\ref{finalVal}) and (\ref{jordanConverge}), we get:
\begin{gather} \label{convergedVal}
\hat{\textbf{x}}_\text{final} = \hat{\mathbf{v}}_{1+} \hat{\mathbf{w}}_{1+}^{T} \hat{\mathbf{x}}(0) = (\hat{\mathbf{w}}_{1+}^{T} \hat{\mathbf{x}}(0)) \hat{\mathbf{v}}_{1+} .
\end{gather}
Since $\hat{A}$ is a row-stochastic matrix, $\hat{\mathbf{v}}_{1+} = \mathbb{1}_{2n}$, proving that the agents reach consensus. Define $\hat{\mathbf{w}}_{1+}^{T} = [\mathbf{w}_{1,\text{1}}^{T}, \mathbf{w}_{1,\text{2}}^{T}]$ . Thus, we get the following eigenvalue problem:
\begin{gather} \label{leftEigs}
\begin{bmatrix}
       \mathbf{w}_\text{1, 1}^T, \mathbf{w}_\text{1, 2}^T 
     \end{bmatrix}
     \begin{bmatrix}\begin{array}{c|c}
       \gamma A & (1-\gamma)A \\[0.5em]
       \hline
       \mathbb{I}_{n \times n} & \mathbb{0}_{n \times n} 
     \end{array}\end{bmatrix} 
 = 1 \begin{bmatrix}
       \mathbf{w}_\text{1, 1}^T, \mathbf{w}_\text{1, 2}^T 
     \end{bmatrix}.
\end{gather}
Solving this system of equations yields:
\begin{gather} \label{leftEigs2}
\mathbf{w}_\text{1, 1}^T A = \mathbf{w}_\text{1, 1}^T .
\end{gather}
Thus, $\mathbf{w}_\text{1, 1}$ is a scaled dominant right eigenvector of $A$ i.e. $\mathbf{w}_\text{1, 1}=\zeta \mathbf{w}_\text{1}$ for some $\zeta\in\mathbb{R}\setminus\{0\}$. Combining this with Eq.~(\ref{leftEigs}) yields $\mathbf{w}_\text{1, 2}= (1-\gamma) \mathbf{w}_\text{1, 1} =  \zeta (1-\gamma) \mathbf{w}_\text{1}$. Additionally, by imposing the constraint $\hat{\mathbf{w}}_{1}^{T}\mathbb{1}_{2n\times1} = 1$, we get:
\begin{gather}
\hat{\mathbf{w}}_{1}^{T} = \begin{bmatrix}
      \left(\frac{1}{2-\gamma}\right)\mathbf{w}_{1}^{T}, \ 
      \left(\frac{ 1-\gamma }{2-\gamma}\right) \mathbf{w}_{1}^{T}
     \end{bmatrix}.
\end{gather}
Hence, the final consensus states attained by agents is:
\begin{gather}\label{x_final}
\hat{\mathbf{x}}_{\text{final}} = \begin{bmatrix}
      \left(\frac{1}{2-\gamma}\right)\mathbf{w}_{1}^{T}, \ 
      \left(\frac{ 1-\gamma }{2-\gamma}\right) \mathbf{w}_{1}^{T}
     \end{bmatrix}\
     \begin{bmatrix}
     \mathbf{x}(0) \\
      \mathbf{x}(0)
     \end{bmatrix} \mathbb{1}_{2n} = \mathbf{w}_{1}^{T} \mathbf{x}(0) \mathbb{1}_{2n}.
\end{gather}

\end{proof}

Furthermore, if $A$ is symmetric, $\mathbf{w}_1 = (1/n)\mathbb{1}_n$ and MLA reaches a consensus to the average of the initial states of the agents.

\section{An accelerated route to consensus}

In the previous section we showed that under certain conditions the MLA algorithm converges while the models by DeGroot and Muthukrishnan et al. do not. In this section we show networks and model parameters for which all three models converge but the suggested algorithm leads to faster convergence. To this end we consider a refinement of Assumption \ref{irreducibleAssume} that guarantees convergence for the DeGroot and the accelerated averaging models, as given below.
\begin{assumption} 
\label{BetterThanDGAssumption}
The weighted adjacency matrix $A$ is \textit{symmetric} and \textit{primitive}, that is $\exists \ k \in \mathbb{N}$ such that $A^k$ is positive.
\end{assumption} 

We define the essential spectral radius of a row-stochastic matrix $A$, denoted by $\rho_{\text{ess}}(A)$, as follows \cite{FB-LNS}.
\begin{gather*} \label{rhoessdefn}
    \rho_{\text{ess}}(A)\equiv \begin{cases} 
      0, \hspace{2.5cm} \text{if spec}(A)=\{1,1,\hdots,1\}, \\
      \max\{|\lambda| \ | \ \lambda \in \text{spec}(A) \setminus \{1\} \}, \hspace{0.5cm} \text{otherwise.} 
   \end{cases}
\end{gather*}
The essential spectral radius determines the rate of convergence for linear discrete models. The lower the essential spectral radius the higher the rate of convergence \cite{FB-LNS}.
\subsection{Comparison with the DeGroot model}

\begin{lemma}
\label{BetterThanDG}
Under Assumptions \ref{NonegRowStochastic} and  \ref{BetterThanDGAssumption}, and if $\lambda_2+\lambda_n \neq 0, \ \exists \ \gamma \in \mathbb{R}$ such that $\rho_{ess}(\hat{A}) < \rho_{ess}(A)$. This implies that the MLA model given by Eq.~(\ref{xhat3}) converges to consensus faster than the DeGroot model given by Eq.~(\ref{DeGroot2}).
\end{lemma}

\begin{proof}
Due to Perron-Frobenius theorem for primitive matrices, (i) $\lambda_1 = 1$ is a simple eigenvalue of $A$, (ii) $|\lambda_i|<1 \ \forall \ i \in \{2,3,\hdots,n\}$, with the convention $1=\lambda_1>\lambda_2\geq \hdots \geq \lambda_n > -1$. For $\gamma = 1 +\Delta$ such that $\Delta \approx 0$, the discriminant in Eq.~(\ref{quadroots}) is approximately equal to $\lambda^2$, and is non-negative. Thus, all eigenvalues of $\hat{A}$ are real i.e. $\hat{\lambda}_{i\pm} \in \mathbb{R} \ \forall \ i \in \{1,2,..,n\}$. Additionally, for $\gamma \approx 1$, the convergence criteria given by Theorem \ref{irreducibleConvCriteria} are satisfied i.e. $A$ is row-stochastic, symmetric, irreducible, $\gamma \in (0,2)$, and $2\gamma\lambda_{n}-\lambda_{n}+1 > 0$. Thus $\hat{A}$ is semi-convergent and reaches consensus. From Eq.~(\ref{quadroots}), we get for $\lambda \neq 0$,
\begin{gather}
\nonumber
\hat{\lambda}_{\pm} = \frac{(1+\Delta)\lambda \pm |\lambda|\sqrt{1 + 2\Delta + \Delta^2 -\frac{4 \Delta}{\lambda} }}{2} 
\end{gather}
Since $\Delta^2 \ll 1$, we ignore this term and use the binomial approximation to get,
\begin{gather}
\hat{\lambda}_{\pm} \approx \frac{(1+\Delta)\lambda \pm |\lambda|\left(1 + \Delta  -\frac{2\Delta}{\lambda} \right)}{2}.
\label{quadroots2}
\end{gather}
We then have,
\begin{gather}
\hat{\lambda}_{i\pm} \approx \left\{\lambda_i + \Delta(\lambda_i - 1), \ \Delta\right\}.
\label{quadroots3}
\end{gather}
Since $\lambda_2 + \lambda_n \neq 0$ there exists a unique eigenvalue $\lambda_{ess}$ such that $\rho_{ess}(A)=|\lambda_{ess}| \in (0,1)$. From Eq.~(\ref{quadroots3}) for $\Delta \approx 0$, we have $\rho_{ess}(\hat{A}) = |\lambda_{ess} + \Delta(\lambda_{ess}-1)|$.
\begin{align}
\nonumber
& \lambda_{ess} \in (0,1) \Rightarrow \ \exists \  \Delta > 0 \ \text{s.t.} \ 0 < \lambda_{ess} + \Delta(\lambda_{ess}-1) < \lambda_{ess}, \\
\nonumber
& \lambda_{ess} \in (-1,0) \Rightarrow \ \exists \  \Delta < 0 \ \text{s.t.} \ 0 > \lambda_{ess} + \Delta(\lambda_{ess}-1) > \lambda_{ess}. \\
& \text{Overall, }  \exists \ \Delta \in \mathbb{R} \ \text{s.t.} \ \rho_{\text{ess}}(\hat{A}) < \rho_{\text{ess}}(A). \label{manymanyinequalities}
\end{align}
\end{proof}
Eq.~(\ref{manymanyinequalities}) is pictorially represented in Fig.~\ref{fig:BetterThanDeGroot}.
The technical assumption $\lambda_2+\lambda_n \neq 0$ is generic and besides pathological cases, will hold for almost all networks. 

\begin{figure}[!ht]
\centering
\begin{subfigure}[b]{0.4\textwidth}         \hspace{0.1cm}\includegraphics[trim= 0 71 25 0, clip,width=\textwidth]{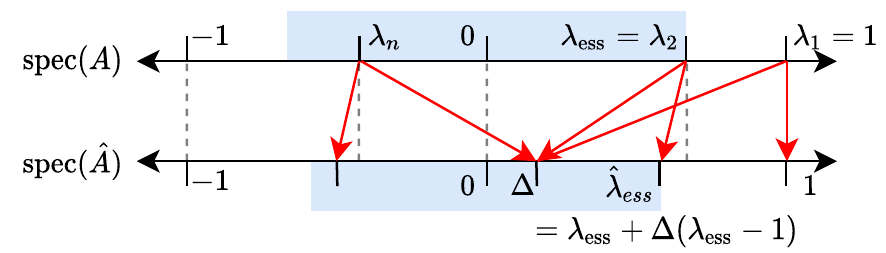}
         \vspace{-0.3in}
         \caption{\small }
     \end{subfigure}
     \begin{subfigure}[b]{0.4\textwidth}
     \hspace{0.1cm}\includegraphics[trim= 0 0 25 71.5, clip,width=\textwidth]{BetterThanDeGroot.pdf}
     \vspace{-0.25in}
     \caption{\small }
     \end{subfigure}
         \caption{\small Eigenvalues of $A$ and their mapping onto those of $\hat{A}$ for $\gamma=1+\Delta$ and $\Delta \approx 0$ for (a) $\lambda_{\text{ess}}=\lambda_2>0$ (b) $\lambda_{\text{ess}}=\lambda_n<0$. The blue shaded regions represent bounds for the non-dominant eigenvalues.}
         \label{fig:BetterThanDeGroot}
\end{figure}
\subsection{Comparison with the accelerated averaging model}

\begin{theorem}  \label{BetterThanAccBullo}
Under Assumptions \ref{NonegRowStochastic}, \ref{BetterThanDGAssumption}, and $\lambda_n<0$, if $\lambda_2 \leq |\lambda_n|/3$, for $\gamma^* = \frac{2}{\rho_{\text{ess}}(A)}\left(\sqrt{1+\rho_{\text{ess}}(A)} - 1\right)$, the essential spectral radius of the MLA model given by Eq.~(\ref{xhat3}) is $\rho_{\text{ess}}(\hat{A}) = \sqrt{1+\rho_{\text{ess}}(A)} - 1$. In this case, the system converges to consensus faster than both the DeGroot model and the accelerated averaging model for any value of the parameter $\beta$.
\end{theorem} 

\begin{proof} 
We will prove this in two parts: first, we show that $\gamma^{*}$ satisfies the convergence criteria in Theorem~\ref{irreducibleConvCriteria}. Second, we compute the convergence rate of the MLA model.

\textit{Part 1}: Under the given assumption, the essential spectral radius of $A$, denoted by  $\rho_{\text{ess}}(A)$, is given by,
\begin{gather}
\rho_{\text{ess}}(A) = -\lambda_n \in (0,1). \label{rhoessA}
\end{gather}
 
We need to verify that $\gamma^{*} \in (0,2)$. 
\begin{gather}
    \sqrt{1+\rho_{ess}(A)}-1 > 0
    \Rightarrow \gamma^*=\frac{2}{\rho_{ess}(A)}\left(\sqrt{1+\rho_{ess}(A)}-1\right)>0
    \label{alphastar1}
\end{gather}

Additionally, 
\begin{align}
\nonumber
&\gamma^*=\frac{2}{\rho_{ess}(A)}\left(\sqrt{1+\rho_{ess}(A)}-1\right)<1 ,\\
    \nonumber
    &\Leftrightarrow 1+ \rho_{ess}(A) < 1+\frac{\rho_{ess}(A)^2}{4} +\rho_{ess}(A), \\
    &\Leftrightarrow \rho_{ess}(A)^2>0.
    \label{alphastar2}
\end{align}
Hence $\gamma^* \in (0,1) \subset (0,2)$

To show that $\gamma^*$ satisfies Theorem~\ref{irreducibleConvCriteria} (ii) we start noting that $\gamma^* \in (0,1)$, as proven before and  $\rho_{ess}(A)>0$, hence
\begin{align}
    \nonumber
    0 < \gamma^* < 1 \Rightarrow \ 2\rho_{ess}(A) > 2\rho_{ess}(A)(1-\gamma^*) > 0, \\
    \nonumber
    \Rightarrow \rho_{ess}(A)+1 > 2\rho_{ess}(A)(1-\gamma^*)-\rho_{ess}(A)+1 > \hdots \\
    \hdots > -\rho_{ess}(A)+1 >0.
\end{align}

Thus, $\lambda_n = -\rho_{ess}(A) \Rightarrow 2\gamma^{*}\lambda_n-\lambda_n + 1 > 0$, as desired for convergence.
Note that the value of $\gamma^*$ is obtained by setting the discriminant in Eq.~(\ref{quadroots}) to zero i.e. $D(\lambda_n, \gamma^*)=0$. From this expression we also get,
\begin{gather}
\ \lambda_n = \frac{4(\gamma^*-1)}{\gamma^{*2}}, \ \text{since} \ \lambda_n < 0.
\label{alphastar02}
\end{gather}
A brief explanation for the choice of $D(\lambda_n,\gamma^*)=0$ is given here. To achieve fastest convergence, we choose a value of the parameter $\gamma$ that minimizes the essential spectral radius of $\hat{A}$. For a chosen value of $\gamma$, the essential spectral radius would equal $\hat{\Lambda}(\lambda, \gamma)=\max\{|\hat{\lambda}_{\pm}(\lambda, \gamma)|\}$, corresponding to a certain value of $\lambda$. A contour plot for the function $\hat{\Lambda}(\lambda, \gamma)$ is given in Fig.~\ref{fig:Dzero}. The solid black curves correspond to the discriminant function, $D(\lambda, \gamma) = 0$, and they divide the parameter space into four quadrants labelled $I-IV$. For each quadrant the sign of $D(\lambda, \gamma)$ is labelled. This sign and that of the parameters, $\lambda$ and $\gamma$, determine which of the two functions $\hat{\lambda}_{\pm}$ has a larger absolute value. For quadrants $I$ and $III$, $D(\lambda, \gamma)<0$ and $\hat{\Lambda}(\lambda, \gamma)=|\hat{\lambda}_{+}(\lambda, \gamma)|=|\hat{\lambda}_{-}(\lambda, \gamma)|=\sqrt{(\gamma-1)\lambda}$. For quadrant $II$, $\hat{\Lambda}(\lambda, \gamma)=|\hat{\lambda}_{+}(\lambda, \gamma)|$. For quadrant $IV$, $\hat{\Lambda}(\lambda, \gamma)=|\hat{\lambda}_{-}(\lambda, \gamma)|$. Furthermore, the directions of monotonic decrease of $\hat{\Lambda}(\lambda, \gamma)$ is indicated by the arrows in the contour plot. Analytically, these directions are determined by the signs of the partial derivatives $\partial \hat{\Lambda}/\partial \lambda$ and $\partial \hat{\Lambda}/\partial \gamma$, where the expression for the function, $\hat{\Lambda}(\lambda, \gamma)$, is known in each quadrant. Let the matrix $A$ have eigenvalues $\{\lambda_1, \lambda_2, \hdots, \lambda_n\}$. From the monotonicity of the function in the contour plot, if $\hat{\Lambda}(\lambda_n, \gamma^*) \geq \hat{\Lambda}(\lambda_2, \gamma^*)$, where $\gamma^*$ is obtained from $D(\lambda_n, \gamma^*)=0$, the parameter value $\gamma^*$ minimizes the function $\hat{\Lambda}(\lambda, \gamma)$ and the thus the essential spectral radius of $\hat{A}$. Thus, we set $D(\lambda_n, \gamma^*)=0$ in order to evaluate the optimizing parameter $\gamma^*$. 

\begin{figure}[!ht]
    \centering
    \includegraphics[width=0.45\textwidth]{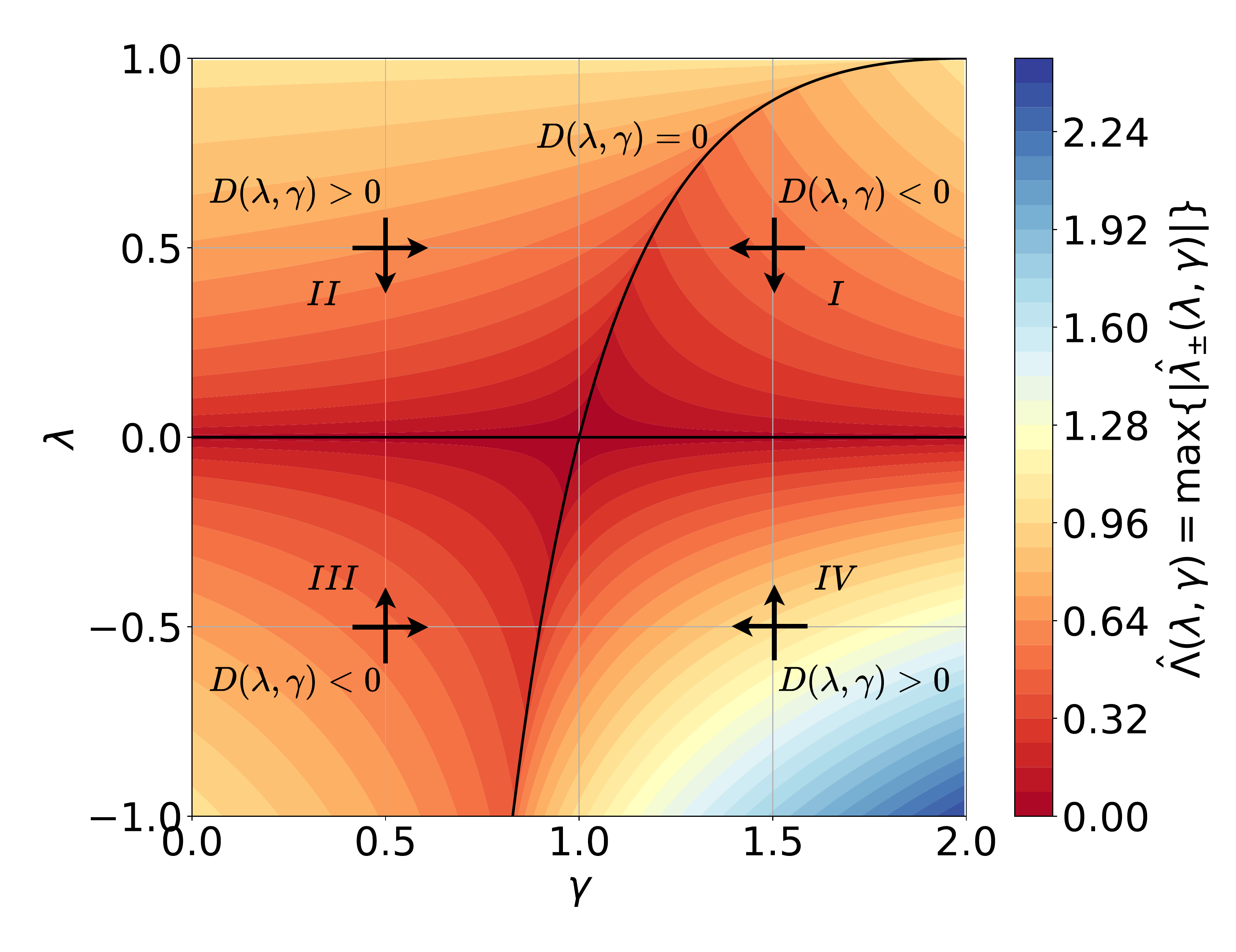} 
    \caption{\small Contour plot of the function $\hat{\Lambda}(\lambda, \gamma)=\max\{|\hat{\lambda}_{\pm}(\lambda, \gamma)|\}$. The solid black curves correspond to the discriminant function, $D(\lambda, \gamma) = 0$, and they divide the parameter space into four quadrants labelled $I-IV$. For each quadrant, the directions of monotonic decrease of the function, $\hat{\Lambda}(\lambda, \gamma)$, and the sign of the discriminant function, $D(\lambda, \gamma)$, are also given.}
    \label{fig:Dzero}  
\end{figure}
\textit{Part 2}: For the convergence rate, we evaluate the essential spectral radius of $\hat{A}$, denoted by $\rho_{\text{ess}}(\hat{A})$, corresponding to $\gamma = \gamma^*$. Note that for $\gamma^* \in (0,1)$, 
\begin{gather}
D(\lambda_n, \gamma^*)=0 \Rightarrow \begin{cases}  D(\lambda_i,\gamma^*) \leq 0 \ \forall \ \lambda_i \in [\lambda_n, 0], \\ D(\lambda_i,\gamma^*) > 0 \ \forall \ \lambda_i \in (0,1]. \end{cases}\label{Dlessthan0}
\end{gather}
We examine the ordering in modulus of eigenvalues of $\hat{A}$. For $\gamma=\gamma^*$, and for some $a$ and $b$, consider two scenarios: (i) $1\geq \lambda_a > \lambda_b > 0$ (ii) $\lambda_n \leq \lambda_a < \lambda_b \leq 0$. In (i), the expression for $\hat{\lambda}_{+}$, given in Eq.~(\ref{quadroots}) is monotonically increasing which implies $ \hat{\lambda}_{a+} > \hat{\lambda}_{b+} > 0$. In (ii), the modulus of complex roots of Eq.~(\ref{quadraticeqn}), is $|(\gamma^{*}-1)\lambda|$ which implies $|\hat{\lambda}_{a\pm}| > |\hat{\lambda}_{b\pm}| > 0$. Thus, for the essential spectral radius of $\hat{A}$, we then have,
\begin{gather}
\rho_{\text{ess}}(\hat{A}) = \max \{|\hat{\lambda}_{1-}|, |\hat{\lambda}_{n\pm}|, |\hat{\lambda}_{2+}| \}, 
\end{gather}
where $|\hat{\lambda}_{1-}|= 1-\gamma^*$, $|\hat{\lambda}_{n\pm}| = \sqrt{(\gamma^*-1)\lambda_n} = -\gamma^*\lambda_n/2$ and $|\hat{\lambda}_{2+}| = \left(\gamma^*\lambda_2 + \sqrt{\gamma^{*2} \lambda_2^2 -4 (\gamma^*-1) \lambda_2}\right)/2$. We show that $|\hat{\lambda}_{n\pm}|> |\hat{\lambda}_{1-}|$. In fact, for  $\gamma^* \in (0,1)$ and using the value of $\lambda_n$ from Eq.~(\ref{alphastar02}), we arrive at the following statements,  
\begin{align} 
\nonumber
\sqrt{(\gamma^*-1)\lambda_n} > |\gamma^*-1|
\nonumber
&\Leftrightarrow \sqrt{(\gamma^*-1)\left(\frac{4(\gamma^*-1)}{\gamma^{*2}}\right)} > |\gamma^*-1|, \\
& \Leftrightarrow -2 < \gamma^* < 2.
\label{rhoessopt}
\end{align}

 which clearly holds $\forall \ \gamma^* \in (0,1)$. Thus, $\rho_{\text{ess}}(\hat{A}) = \max \{|\hat{\lambda}_{n\pm}|, |\hat{\lambda}_{2+}| \}$. We next show that $\rho_{\text{ess}}(\hat{A}) = |\hat{\lambda}_{n\pm}|$ if and only if
\begin{align}
\nonumber
&|\hat{\lambda}_{n\pm}| \geq |\hat{\lambda}_{2+}|,\\
\nonumber
 & \Leftrightarrow -\frac{\gamma^*\lambda_n}{2} \geq \frac{\gamma^*\lambda_2 + \sqrt{\gamma^{*2} \lambda_2^2 -4 (\gamma^*-1) \lambda_2}}{2}, \\
\nonumber
\nonumber
&\Leftrightarrow -\lambda_n - \lambda_2 \geq   \sqrt{\lambda_{2}^{2} - \lambda_{n}\lambda_{2}},\\
&\Leftrightarrow c_1: \lambda_n + \lambda_2  \leq 0 \ \text{and} \ c_2: |\lambda_n + \lambda_2| \geq  \sqrt{\lambda_{2}^{2} - \lambda_{n}\lambda_{2}}. \label{rhoEssEqn}
\end{align}
Condition $c_1$ holds for $\lambda_n <0$ and $|\lambda_n|\geq |\lambda_2|$. Condition $c_2$ in Eq.~(\ref{rhoEssEqn}) can be further simplified as shown below:
\begin{align}
\nonumber
c_2 &\Leftrightarrow (\lambda_n + \lambda_2)^2 \geq  \lambda_{2}^{2} - \lambda_{n}\lambda_{2}, \\
\nonumber
     & \Leftrightarrow \lambda_n(\lambda_n + 3\lambda_2) \geq 0,\\
     \nonumber
    & \Leftrightarrow \lambda_n \leq -3\lambda_2, \ \text{since} \ \lambda_n < 0, \\
    &  \Leftrightarrow \lambda_2 \leq \frac{|\lambda_n|}{3}. \label{most_restrictive}
\end{align}
Overall, the expression for the spectral radius of $\hat{A}$ is then given by,
\begin{gather}
\rho_{\text{ess}}(\hat{A})  =|\hat{\lambda}_{n\pm}|=-\gamma^*\lambda_n/2. \label{rho_ess_ahat_00}
\end{gather}
Substituting Eq.~(\ref{rhoessA}) and $\gamma^*$ from Theorem \ref{BetterThanAccBullo} in Eq.~(\ref{rho_ess_ahat_00}), we get:
\begin{gather}
    \rho_{\text{ess}}(\hat{A}) = \sqrt{1+\rho_{\text{ess}}(A)} - 1. \label{rhoessAhat2}
\end{gather}
The optimized essential spectral radius of the accelerated averaging model given by Eq.~(\ref{muthuModel}) is given by \cite{FB-LNS},
\begin{gather}
\min_{\beta\in(0,2)}\rho_{\text{ess}}(A_\beta)=\rho_{\text{ess}}(A_{\beta^*})= \frac{\rho_{\text{ess}}(A)}{1+\sqrt{1-\rho_{\text{ess}}^2(A)}} < \rho_{\text{ess}}(A), \label{accBullo}
\end{gather}
for $\rho_{\text{ess}}(A) \in (0,1)$.
Using further algebraic manipulations it can then be shown that, for $\rho_{\text{ess}}(A) \in (0,1)$, MLA converges faster than the accelerated averaging model and DeGroot since,
\begin{gather}
\rho_{\text{ess}}(\hat{A}) = \sqrt{1+\rho_{\text{ess}}(A)} - 1 <\rho_{\text{ess}}(A_{\beta^*}) < \rho_{\text{ess}}(A). \label{finalineq}
\end{gather}
\end{proof} 
\vspace{-0.5cm}
\begin{figure}[!ht]
    \centering
    \includegraphics[clip, width=0.3\textwidth]{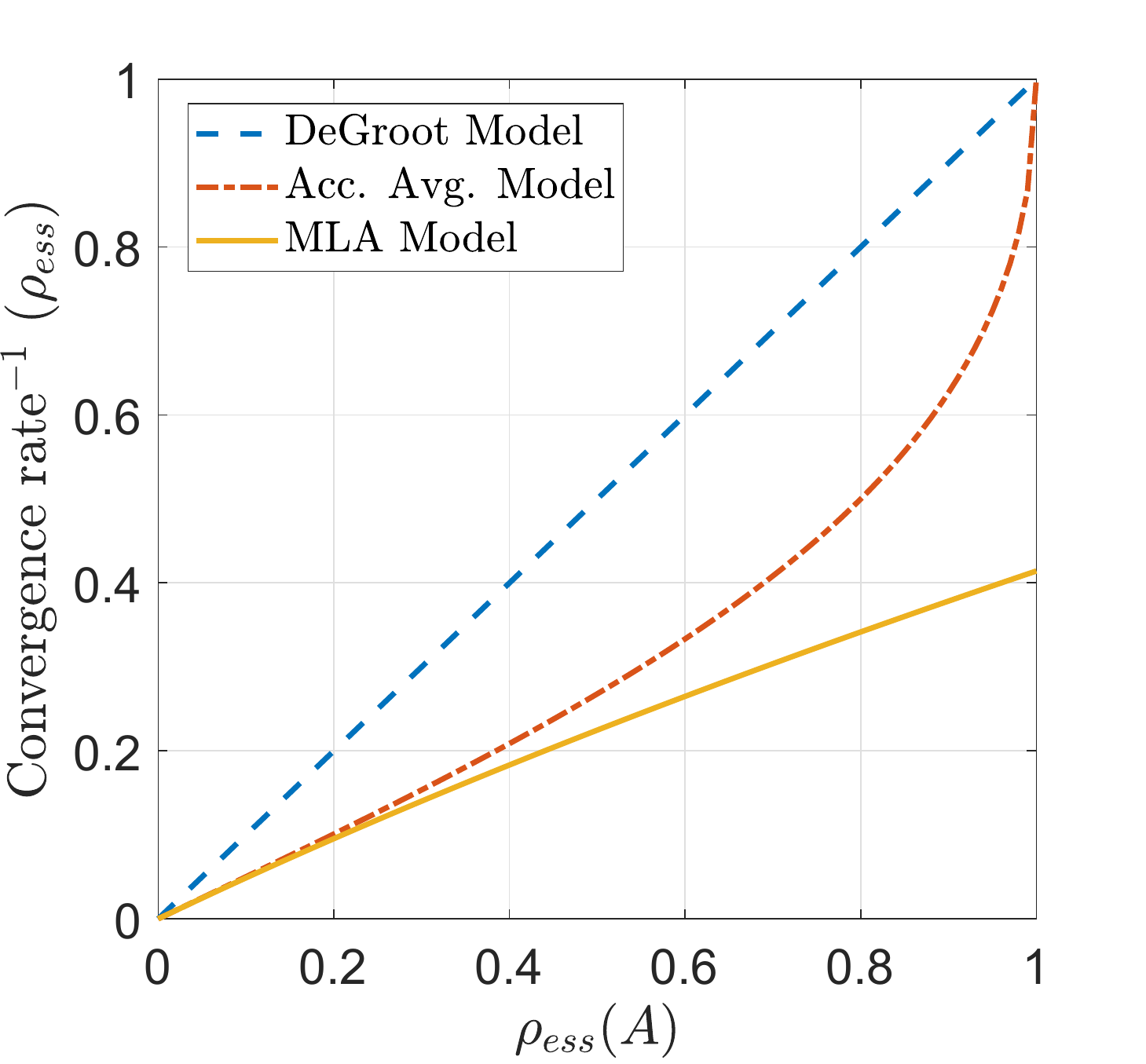} 
    \caption{\small Comparison of essential spectral radii and the optimal rates of convergence of the three models of network dynamics; DeGroot model, accelerated averaging model, and the MLA model under the criteria given in Theorem \ref{BetterThanAccBullo}. MLA model yields faster convergence.}
    \label{fig:rateOfConvergence}  
\end{figure}

Eq.~(\ref{finalineq}) is plotted in Fig.~\ref{fig:rateOfConvergence}. It can be noted that for a network that satisfies the criteria given in Theorem \ref{BetterThanAccBullo} with an essential spectral radius close to unity, the MLA model reaches convergence value significantly faster than the classic DeGroot model and the accelerated averaging model. Good candidates for such matrices are perturbations of a periodic network. A schematic for such networks is shown in Fig.~\ref{fig:AlmostPeriodic}. In such networks, $\lambda_n >-1$ but such that $\rho_{ess}(A) \approx 1$. One such perturbation is shown in Fig.~\ref{fig:fastestCoverge} which is obtained by adding self-loops of low weight in the 4-node ring network. The corresponding time series for the three models with optimized parameters are also shown.
\begin{figure}
     \centering

       \vspace{-0.3cm} \includegraphics[trim= 0 35 0 5, clip, width=0.35\textwidth]{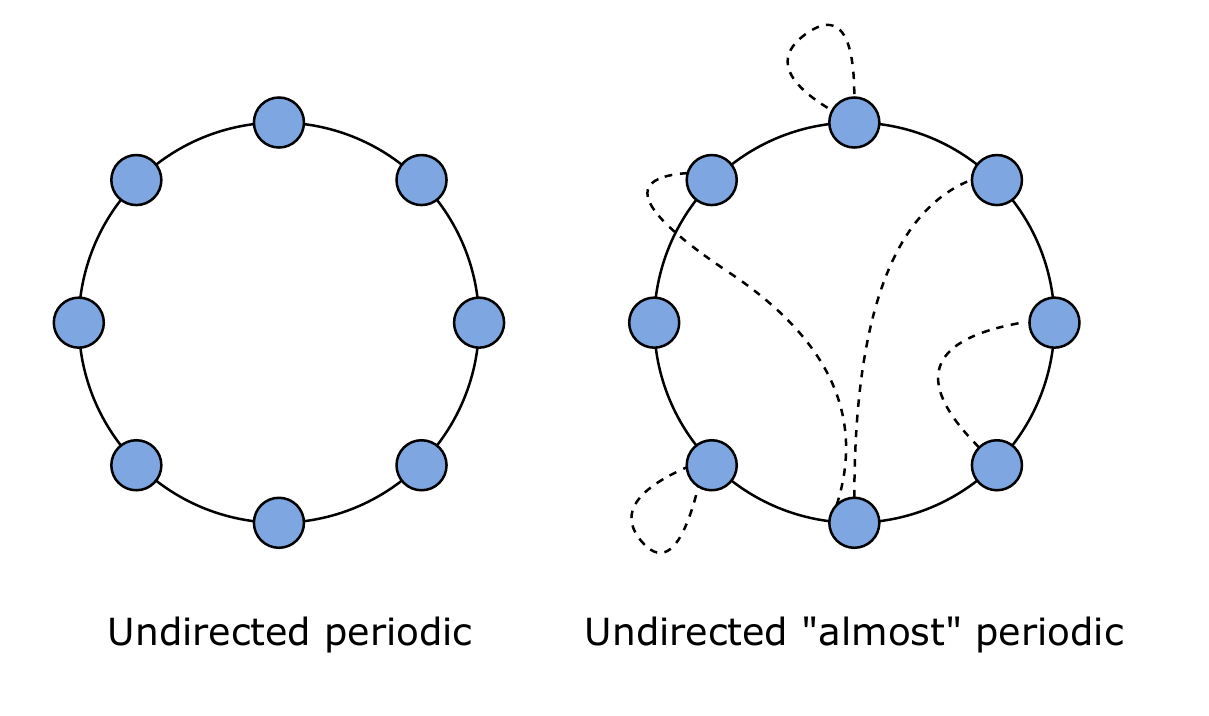}
         \caption{\small A periodic ring network is shown on the left. Network links with low weights, represented by the dashed curves, are added to this periodic network while maintaining row-stochasticity of the corresponding adjacency matrix to obtain candidates for networks that satisfy criteria given in Theorem~\ref{BetterThanAccBullo}.}
         \label{fig:AlmostPeriodic}
\end{figure}

\begin{figure}
     \centering

        \includegraphics[trim= 50 10 40 30, clip, width=0.4\textwidth]{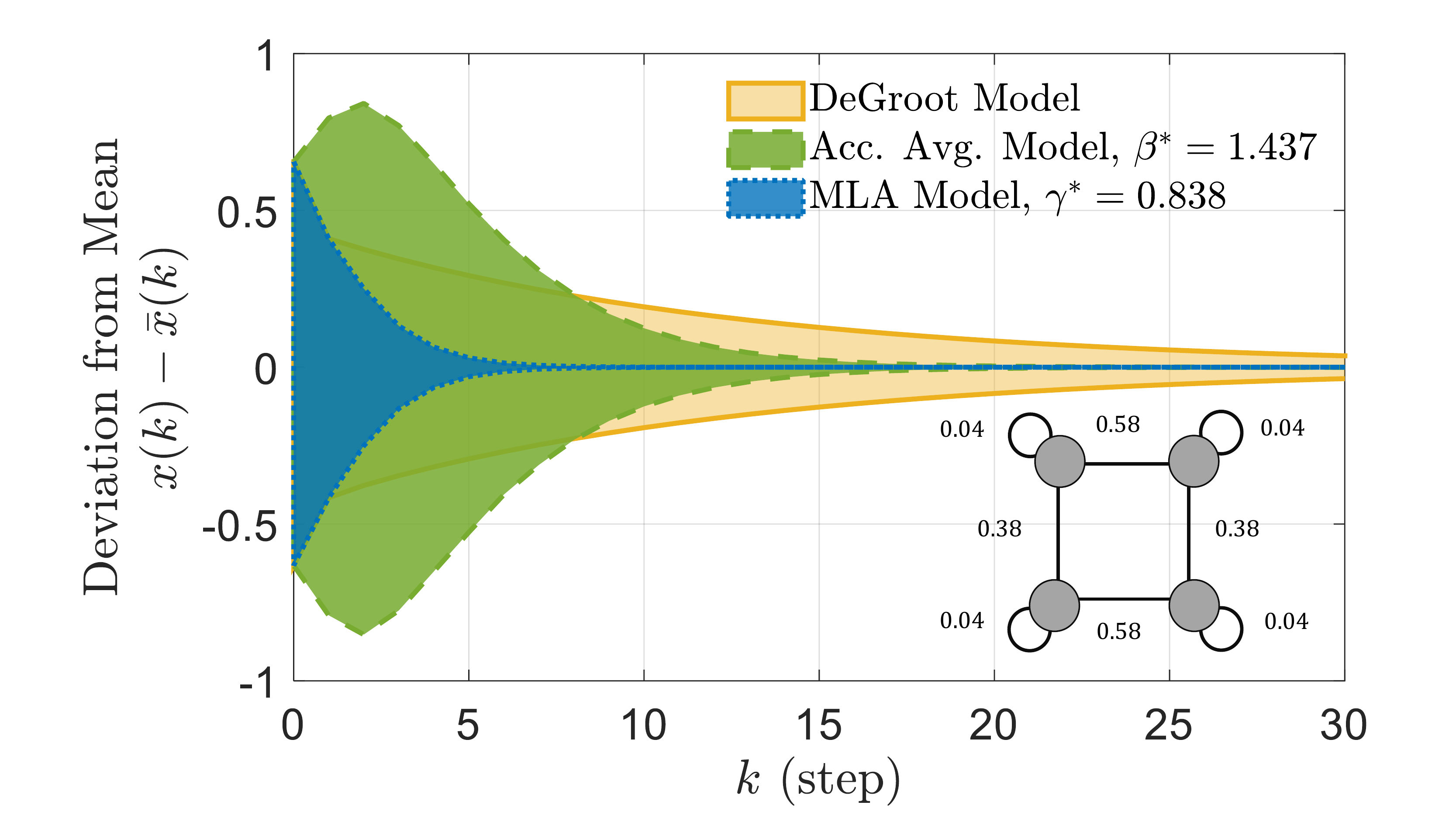}
         \caption{\small 
         Comparison of rates of convergence in three averaging models for a ring network with self-loops. The plot shows the envelope enclosed between the maximum and minimum \textit{deviation} from the mean state and is simulated using 1000 randomly chosen initial conditions. MLA converges faster.}
         \label{fig:fastestCoverge}
\end{figure}

\section{Conclusion}

We propose a simple modification to the accelerated averaging scheme introduced in \cite{muthukrishnan1998first}. In the proposed model (MLA), we apply the DeGroot update to the the current and past states of neighboring nodes followed by a linear combination step. The MLA model is applicable to networks that are symmetric, primitive and row-stochastic. We find the optimal model parameter, $\gamma^*$, for which MLA converges faster than both the DeGroot and accelerated averaging algorithms under certain network constraints. Another important contribution is that unlike the other two algorithms, MLA converges even for periodic networks. A summary of the important results is given in Fig.~\ref{fig:Summary}. Future work could involve extending the results presented to a larger class of matrices such as asymmetric matrices.
\begin{figure*}
     \centering
\vspace{-0.35cm}
        \includegraphics[trim= 0 0 0 0, clip, width=0.85\textwidth]{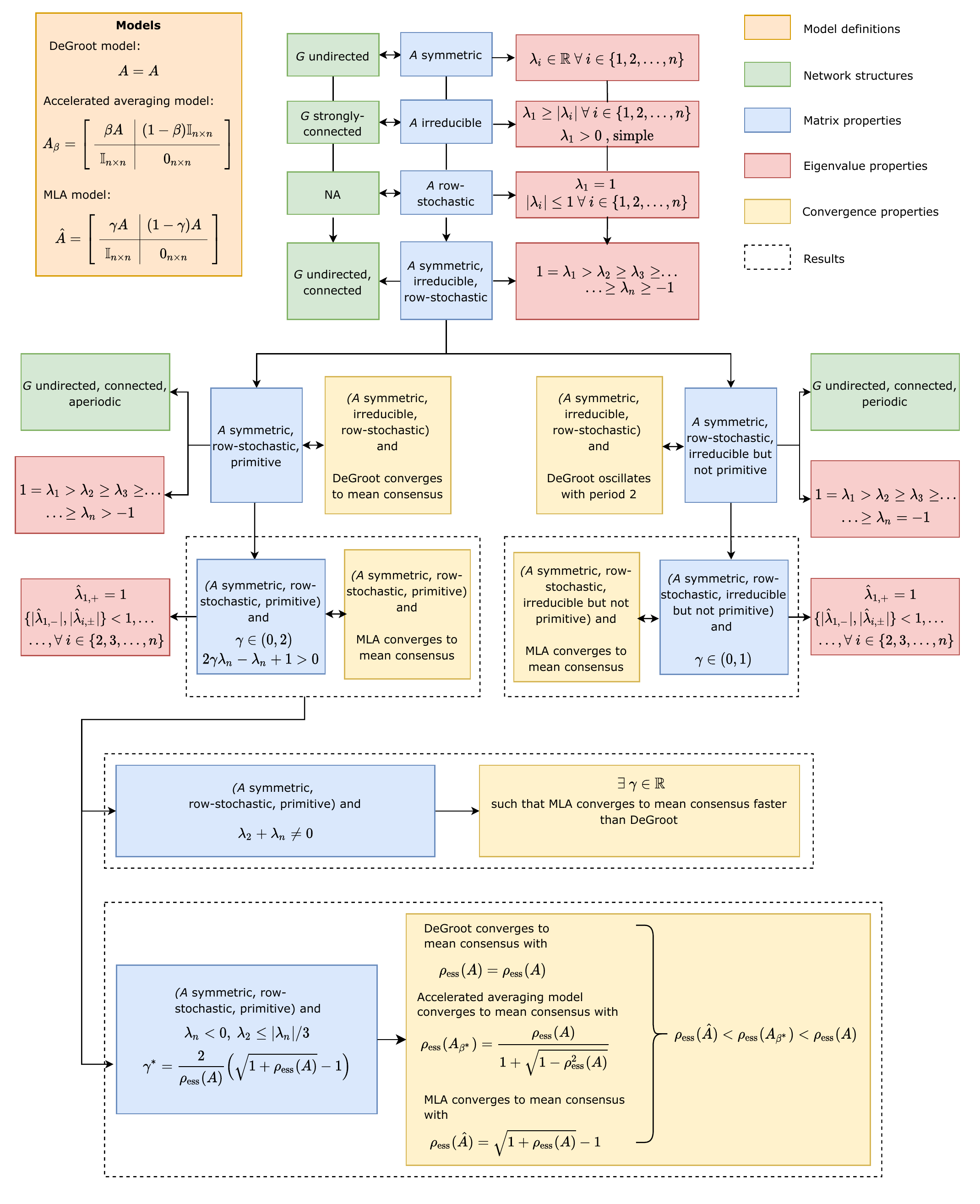}
         \caption{\small A summary of results for MLA model. Here, $G$ is the network corresponding to the adjacency matrix $A$ for the DeGroot model. All other notations follow from the main text.}
         \label{fig:Summary}
\end{figure*}

\bibliographystyle{IEEEtran} 

\end{document}